\documentclass[12pt]{article}

\usepackage{setspace}

\usepackage[left=-1in,right=-1in]{geometry}

\usepackage{amsfonts,epsfig,fullpage}
\usepackage{amsmath, amsfonts, url}
\usepackage[table]{xcolor}

\begin{document}

\newtheorem{theorem}{Theorem}[section]
\newtheorem{lemma}[theorem]{Lemma}
\newtheorem{fact}[theorem]{Fact}
\newtheorem{claim}[theorem]{Claim}
\newtheorem{corollary}[theorem]{Corollary}
\newtheorem{conjecture}[theorem]{Conjecture}
\newtheorem{question}[theorem]{Question}
\newtheorem{proposition}[theorem]{Proposition}
\newtheorem{axiom}[theorem]{Axiom}
\newtheorem{remark}[theorem]{Remark}
\newtheorem{example}[theorem]{Example}
\newtheorem{exercise}[theorem]{Exercise}
\newtheorem{definition}[theorem]{Definition}
\newtheorem{observation}[theorem]{Observation}
\newtheorem{protocol}{Protocol}
\newtheorem{recipe}{Recipe}

\def\pproof{\par\penalty-1000\vskip .5 pt\noindent{\bf Proof\/ }}
\newcommand{\QED}{\hfill$\;\;\;\rule[0.1mm]{2mm}{2mm}$}

\newenvironment{proof}{\begin{pproof}}{\QED\end{pproof}~\\}

\newcommand{\ExactT}{$f^G_{k}$~}
\newcommand{\Exactn}{{\rm Exactly-}$n$~}

\newcommand{\ind}{\ensuremath{\alpha}}
\newcommand{\chr}{\ensuremath{\chi}}

\newcommand{\rs}{Ruzsa-Szemer\'{e}di}

\title{Algorithmic Number On the Forehead Protocols Yielding 
Dense \rs\ Graphs and Hypergraphs}

\author{Noga Alon 
\thanks{Research supported in part by
NSF grant DMS-1855464, ISF grant 281/17,
BSF grant 2018267
and the Simons Foundation.}\\
Princeton University\\
Princeton, NJ 08544, USA\\
and Tel Aviv University \\
Tel Aviv 69978, Israel \\
{\tt nalon@math.princeton.edu}
\and Adi Shraibman \\
   	The Academic College of Tel-Aviv-Yaffo \\
    	Tel-Aviv, Israel \\
    	{\tt adish@mta.ac.il}
}

\date{}

\maketitle

\abstract{

We describe algorithmic Number On the Forehead protocols that 
provide dense \rs\ graphs. 
One protocol leads to a simple and natural extension of the original 
construction of Ruzsa and Szemer\'{e}di. 
The graphs induced by this protocol have $n$ vertices, $\Omega(n^2/\log n)$ 
edges, and are decomposable into $n^{1+O(1/\log \log n)}$ induced matchings.
Another protocol is an explicit (and slightly simpler) version of the 
construction of \cite{AMS}, 
producing graphs with similar properties. 
We also generalize the above protocols to more than three players, in order 
to construct dense uniform hypergraphs in which every edge lies in a positive
small number of simplices.
}

\section{Introduction}

For an integer $n$ and a positive real $c$, let $h(n,c)$ denote the
maximum number so that any $n$ vertex graph with at least $cn^2$
edges in which every edge is contained in a triangle, must contain
an edge lying in at least $h(n,c)$ triangles.  Erd\H{o}s and
Rothschild asked to determine or estimate $h(n,c)$,
see \cite{CG}, \cite{Er1}, \cite{Er2}, \cite{Er3}. Szemer\'edi
observed that the triangle removal lemma (see \cite{RSz})
implies that for every
fixed $c>0$, $h(n,c)$ tends to infinity with $n$, and Trotter and the
first author noticed that for any $c<1/4$ there is a $c'$ so that
$h(n,c) <c' \sqrt n$. A clever construction of Fox and Loh
\cite{FL} shows that in fact for any fixed $c<1/4$ , $h(n,c) \leq
n^{O(1/\log \log n)}$. While this is still very far from the lower
bound based on the triangle removal lemma and its improved 
quantitative version in \cite{Fo}, which provides a lower bound
exponential in $\log^*n$ for any fixed $c>0$, it does show that
$h(n,c) =n^{o(1)}$. Note that the constant  $1/4$ is tight, as it
is known that any $n$-vertex graph with $\lfloor n^2/4 \rfloor+1$
edges must contain an edge lying in at least $n/6$ triangles (see
\cite{HN}).

The construction of Fox and Loh triggered another surprising result
in the study of a closely related problem. The first author, 
Moitra and Sudakov \cite{AMS}
constructed $(r,t)$-\rs\ graphs on $n$ vertices 
with $r = n^{1-o(1)}$ and $rt = (1-o(1))\binom{n}{2}$.
A graph is an $(r,t)$-\rs\ graph if its set of edges can be partitioned into $t$ pairwise disjoint induced matchings, each of size $r$.
These graphs were introduced in a paper by Ruzsa and Szemer\'{e}di \cite{RSz}. They used these graphs, together with the
regularity lemma of Szemer\'{e}di \cite{Szem} to tackle the 
so called $(6,3)$-problem dealing with the maximum possible number of
edges of a $3$-uniform hypergraph on $n$ vertices that contains no
$3$ edges spanning at most $6$ vertices. 
\rs\ graphs have been studied extensively since, finding applications in Combinatorics, Complexity Theory and Information Theory.
A natural line of research is to find dense graphs with 
relatively large $r$. One such construction is given by Birk, Linial 
and Meshulam \cite{birk1993uniform}, 
with $r=(\log n)^{\Omega(\log\log n/ (\log\log\log n)^2)}$ 
and $t = \Omega(n^2/r)$.
Meshulam conjectured that there are no $(r,t)$-\rs\ graphs with 
both $rt = \Theta(\binom{n}{2})$ and $r \geq n^{\Omega(1)}$.  
The construction from \cite{AMS} disproved Meshulam's conjecture in a strong
form, vastly improving the one in \cite{birk1993uniform}. 

The first aim of the present short paper 
is to describe these results in communication complexity terms by providing
algorithmic Number-On-the-Forehead (NOF, for short) 
protocols that entail them. 
\rs\ graphs are closely related to the NOF model in communication complexity, as observed in \cite{hdp17}. 
They are related to the communication complexity of $2$-dimensional permutations and sub-permutations
(see details in the sequel).
We observe here that communication protocols in the NOF model for $2$-dimensional permutations also 
imply upper bounds on $h(n, c)$.
 
\iffalse
In broad terms, the question of finding dense \rs\ graphs can be stated as follows in the language of communication 
complexity: Fixing $c \in \mathbb{N}$, what is the minimal $N$ (as a function of $n$) so that there exists a sub-permutation $f: [n]^2 \times [N] \to \{0,1\}$
with communication complexity at most $c$? Alternatively, fixing $N$, what is the minimal communication complexity of a sub-permutation $f: [n]^2 \times [N] \to \{0,1\}$?
\fi

We give algorithmic NOF protocols that derive the constructions of dense \rs\ graphs from \cite{AMS} and also
the results of Fox and Loh \cite{FL}. 
This makes the constructions strongly explicit and also somewhat simpler. Another advantage of this approach
is that it provides a clear link between these results and the original results of Ruzsa and Szemer\'{e}di \cite{RSz}. 

The second aim of this paper is to extend the above mentioned 
results to uniform hypergraphs.
To do so we extend the protocols to any number $k >3$ of players.  
Let $K_k=K_k^{(k-1)}$ denote the complete $(k-1)$-uniform hypergraph
($(k-1)$-graph, for short) on $k$ vertices.
For an integer $n$ and a positive real $c$, let $h_{k-1}(n,c)$ denote the
maximum number so that any $n$ vertex 
$(k-1)$-graph with at least $cn^{k-1}$
edges, in which every edge is contained in a copy of $K_k$,
must contain
an edge lying in at least $h_{k-1}(n,c)$ such copies. By the hypergraph
removal lemma proved in \cite{Go} and independently in
\cite{RS}, \cite{NRS}, for any fixed positive $c$, $h_{k-1}(n,c)$ tends
to infinity with $n$. Indeed, for example, if $G$ is an $n$-vertex $3$-graph
with at least $cn^3$ edges, and each edge  is contained in at least
$1$ and at most $h=h_3(n,c)$ copies of $K=K_4$, then $G$ must contain
at least $\frac{cn^3}{4h}$ pairwise edge-disjoint copies of $K$.
Hence at least that many edges have to be omitted from $G$ in
order to destroy all copies of $K$, and thus by the hypergraph 
removal lemma if $h$ is a constant then
$G$ must contain at least $\Omega(n^4)$ copies of 
$K$, implying that some edges are contained in $\Omega(n)$ such
copies, contradiction.

Unlike the graph case, the maximum possible number $ex_{k-1}(n,K_k)$ 
of edges of an
$n$-vertex 
$(k-1)$-graph with no copies of $K_k$ is not known. The determination of
this number is an old
problem posed by Tur\'an \cite{Tu}, and Erd\H{o}s offered a
significant award for its solution, see
\cite{Er0}. By a  general result 
proved in \cite{KNS}, the limit of the ratio
$$
\frac{ex_{k-1}(n,K_k)}{n^{k-1}}
$$
as $n$ tends to infinity exists. This is a positive number  
called the Tur\'an density of $K_k$. Let $d_k=d(K_k)$ denote this number,
which is conjectured to be $5/9$ for k=4. 
See \cite{CL} and its references for some of the work on this
problem.  Although $d_k$ is not known, we can prove the
following.
\begin{theorem}
\label{t11}
For any fixed $c<d_k$ there is some $b>0$ so that
$h_{k-1}(n,c) \leq n^{b/ \log \log n}$. 
\end{theorem}

Note that by the results of \cite{ES} on supersaturated
hyperghraphs if $c>d_k$ then any $(k-1)$-graph on $n$ vertices with 
at least $cn^{k-1}$ edges contains $\Omega(n^k)$ copies of $K_k$.
Therefore, for any such $c$ there is a constant $b=b(c)>0$ so that
$h_{k-1}(n,c) \geq b n$, implying that the $d_k$ bound in Theorem
\ref{t11} is tight.

Our protocols also imply an extension of the main result of \cite{AMS}.
That is, it entails a construction of nearly complete $(k-1)$-graphs whose edges can
be partitioned into a nearly linear number of induced subgraphs,
each being a partial Steiner system. Recall that a
$(k-1)$-graph is a partial Steiner system if no two of its edges
share $k-2$ common vertices. It is clear that any such graph on 
$n$ vertices cannot contain more than $\frac{1}{k-1}{n \choose k-2}<n^{k-2}$
edges, and hence any $(k-1)$-graph with at least $bn^{k-1}$ edges cannot be
partitioned into less than $\Theta(n)$ partial Steiner
systems. The hypergraph removal lemma shows here, too, that 
in fact the number of such systems cannot be $\Theta(n)$, that is,
for any fixed positive $b$,
this number divided by $n$ must tend to infinity with $n$. The
following result shows, however, that this number can be smaller
than $n^{1+\epsilon}$ for any positive $\epsilon$.

\begin{theorem}
\label{t12}
For every integer $k \ge 3$, there is an absolute constant $c>0$ so that for sufficiently
large $n$ there is a $(k-1)$-graph on $n$ vertices with at least
$$
(1-o(1)){ n \choose k-1}
$$
edges, whose edges can be decomposed into at most
$n^{1+c/\log \log n}$ induced subgraphs, each being a partial
Steiner system.
\end{theorem}

The rest of the paper contains the proofs of the above two
theorems. The organization is as follows. Section~\ref{sec:background} contains
background on communication complexity and high-dimensional permutations, a recipe for proving 
Theorem~\ref{t11} and Theorem~\ref{t12} using communication protocols,
and a simple application of this recipe to construct a graph on $n$ vertices and $\Omega(n^2/\log n)$ 
edges, decomposable into $n^{1+O(1/\log \log n)}$ induced matchings.
Section~\ref{apply_recipe_to_prove_theorems} contains the application of this recipe to prove Theorem~\ref{t11} and Theorem~\ref{t12}.
The details of the graphs and hypergraphs produced by this recipe, and the proof 
that it works correctly  are given in 
Section~\ref{sec:recipe_proof_complete}. The final Section \ref{s5} 
contains
a brief summary.

\iffalse
Sections~\ref{sec:background}-\ref{graphs} gives the communication 
complexity interpretation of the results of Fox and Loh \cite{FL} and  
of the first author, 
Moitra and Sudakov \cite{AMS}. Section~\ref{sec:background} explains
how communication protocols can imply such results, and in the rest of the sections
we describe the necessary protocols and analyze them. Section~\ref{alternative} provides
the basic function and protocol that are used later on, and shows how they can be used to derive an alternative protocol
for \Exactn which is interesting in its own right. 
Recall that protocols for \Exactn
are intimately related to corner theorems over the integer grid (see e.g. \cite{hdp17} for details).
Hence, this approach gives a nice link between the new constructions of \rs\ graphs and the 
original construction of Ruzsa and Szemer\'{e}di \cite{RSz}, which essentially follows 
from the known protocol for \Exactn based on Behrend's construction \cite{Be}.  
In Section~\ref{graphs} we first describe
a very simple protocol that already gives very dense \rs\ graphs. Section~\ref{The_constructions_of_Alon_etal} gives the 
results of Alon, Moitra and Sudakov. 
Section~\ref{hypergraphs} contains the generalization
to uniform hypergraphs, and thus also the proofs of Theorems~\ref{t11} and \ref{t12}.
\fi

\newcommand{\trans}{T}

\section{From communication to graphs and hypergraphs}
\label{sec:background}

\subsection{Background and notation}

\paragraph{General notation}

We let $[n] = \{1,2,\ldots,n\}$. A $k$-tuple is denoted either $(x_1,\ldots,x_k)$ or
in abbreviated form $\vec{x}$. 

\paragraph{Communication complexity}

We start with a few basic communication complexity notions. The definitions we give are a simplified version and 
adjusted to our needs. The interested reader can see \cite{KN97} for a more comprehensive survey.
In the NOF model $k$ players wish to compute a function $f:X_1\times X_2 \times \cdots \times X_k \to \{0,1\}$.
The players agree on a communication {\em protocol} $P$.
Then, an input $(x_1, x_2,\ldots,  x_k)$ is presented to the players so player $i$ sees all input except $x_i$,
we sometimes refer to this player as the $x_i$-player.
The players take turns to write messages on a blackboard according to the agreed protocol $P$. 
Each message of each player may depend on the part of the input seen by this player,
and except for the last player it can also depend on the messages written so far on the blackboard.
The message written by the last player depends only on the part of the input he sees, and is independent
of the content of the blackboard.
One way to visualize this is as if the last player wrote a message first and then did not participate in the rest
of the transaction. The value of the function can be computed by all
players from the content of the board at the end of the protocol.
The {\em cost} of a protocol, denoted $C(P)$, is the maximal 
number of bits written on the board, over all inputs, by the first $k-1$ players
\footnote{In the basic communication complexity definition all players can see each others messages, and
the cost of the protocol depends also on the message of the last player. The version of communication complexity we
gave here is from the {\em one-sided} model. Since we only need this version, we simplify our notations.}.

The string of bits written on the blackboard for a given input $\vec{x}=(x_1,\ldots,x_k)$ is called a {\em transcript},
denoted $\mathcal{T}(\vec{x})$. We let
$\mathcal{T}_i(\vec{x})$ for $i=1,\ldots,k$ be the part of this transcript that is written by player $i$.
Let $\trans$ be a transcript, the subset $S=S(\trans)$ 
of entries satisfying  
$\mathcal{T}(\vec{x}) = \trans$ and $f(\vec{x}) = 1$, 
is called a {\em cylinder intersection} \footnote{The usual definition of cylinder intersection is more general, 
what we defined here is 
referred to as a $1$-monochromatic cylinder intersection. Since we are only interested in $1$-monochromatic cylinder intersections
we abbreviate the notation.}.  Note that a cylinder intersection is defined with respect to a function and a protocol for this function,
we specify the function and protocol when it is necessary for a clear presentation and otherwise omit them.

We say that a subset of entries $S$ is {\em symmetric}
if membership in $S$ does not depend on the order of the first $k-1$ entries. That is, $S$ is symmetric if
$(x_1,\ldots,x_{k-1},x_k) \in S$ if and only if $(x_{\pi(1)},\ldots,x_{\pi(k-1)},x_k) \in S$ for every 
permutation $\pi$ on $\{1,2,\ldots,k-1\}$.

\paragraph{High-dimensional permutations}

A {\em line} in $[n]^{k}$ is a subset $L \subset [n]^k$ such that $k-1$ of the coordinates in $L$
are fixed, and the remaining coordinate takes all possible values. Following is a simple example with $n=5$ and $k=3$:
\[
L = \{(1, 1, 4), (1,2,4), (1,3,4), (1,4,4), (1,5,4)\}.
\] 
In this example the first and third coordinates are fixed, and the second coordinate takes all possible values in $[5] = \{1,2,3,4,5\}$.
There is a distinct line for every choice of unconstrained coordinate $i \in [k]$, and a choice of values to fix the 
remaining coordinates. A line in $[n_1]\times \cdots \times [n_k]$ is defined similarly. We say that 
the line is in the {\em $i$th dimension} if the unconstrained coordinate is $i$. 

A {\em $(k-1)$-dimensional permutation} is a function $f:[n]^k \to \{0,1\}$ such that for every line $L$
in $[n]^k$ there is exactly one $\vec{x} \in L$ such that $f(\vec{x})=1$.
A {\em sub-permutation} is a function $f:[n]^{k-1}\times [N] \to \{0,1\}$ such that every line in the $k$th dimension
contains a single $1$, and every other line contains at most
one $1$. 

For example, let $G$ be a group, define $f: G^k \to \{0,1\}$ 
by $f(x_1, \ldots, x_k)=1$ if and only if
$x_1+x_2+\cdots +x_k = 0$. Then $f$ is a permutation. 
\iffalse
For example, let $G$ be a quasigroup, define $f: G^k \to \{0,1\}$ 
by $f(x_1, \ldots, x_k)=1$ if and only if
$x_1+x_2+\cdots +x_k = 0$. Then $f$ is a permutation. 
In fact this characterizes all permutations.
\fi
Let $H$ be a subset of $G$, then the function $h: H^{k-1}\times G \to \{0,1\}$ defined 
similarly to $f$, is a sub-permutation.  

A {\em weak permutation} is a function $f:[n]^k \to \{0,1\}$ such that every line
contains at most one $1$-entry, and a {\em weak sub-permutation} 
is defined similarly: it is an $f:[n]^{k-1} \times [N] \to \{0,1\}$
with $N \geq n$ such that every line contains at most one $1$-entry.

\paragraph{\rs\ graphs and hypergraphs}

As mentioned in the introduction, a graph is an $(r,t)$-\rs\ graph if its set of edges can be partitioned into 
$t$ pairwise disjoint induced matchings, each of size $r$. Such a graph obviously has $rt$ edges. A challenge
in constructing \rs\ graphs is to make the density of edges as large as possible while keeping the number of 
matchings relatively low. We are therefore less concerned with the size of each matching, and only worry about the number 
of matchings and the density of the edges.

There is a natural way to extend the notion of \rs\ graphs to hypergraphs, by considering 
Steiner systems $S(k-2,k-1)$. A {\em Steiner system} $S(t-1,t)$ 
in a set $V$, is a family of $t$-element subsets of $V$ (called {\em blocks}) 
such that each $(t-1)$-element subset of $V$ is contained in exactly one block. A {\em partial Steiner
system} is defined similarly with the exception that each $(t-1)$-element subset of $V$ is contained 
in {\bf at most} one block. 

For a natural number $k > 2$, and a $(k-1)$-graph $G=(V, E)$ we 
are interested in partitioning $E$ into
induced partial Steiner systems $S(k-2,k-1)$. Note that if $V$ is the set of vertices of a graph, then 
a partial Steiner system $S(1,2)$ in $V$ is a matching. Thus, this definition extends the notion of a \rs\ graph.

\subsection{A recipe}
\label{subsec:a_recipe}

Given a function $f:[n]^{k-1} \times [N] \to \{0,1\}$, a protocol $P$ for $f$,
and a transcript $\trans$ of the last player, denote
\[
S_{k}(\trans) = \{ (x_1,\ldots,x_k) \in [n]^{k-1} \times [N]: \mathcal{T}_k(x_1,\ldots,x_k) = \trans 
\;\; and \;\; f(x_1,\ldots,x_k)=1 \}. 
\]
Next we describe a recipe for generating \rs\ graphs and hypergraphs, as well as upper bounds on $h_{k-1}(n, c)$, 
from NOF protocols. 

\bigskip
\fbox{
\begin{minipage}{5.8in}
\begin{recipe}{\sf{- from protocols to graphs and hypergraphs}}
\label{recipe}
\begin{enumerate}

\item Choose a weak sub-permutation $f:[n]^{k-1} \times [N] \to \{0,1\}$,
for natural numbers $n$, $N$ and $k>2$.

\item Construct a communication protocol $P$ for $f$.

\item Pick a transcript $\trans$ 
of the last player so that $S_k(\trans)$
is symmetric, and let $S = S_k(\trans)$. 
\end{enumerate}

\end{recipe}
\end{minipage}
}
\bigskip

The following theorem describes the outcome when following Recipe~\ref{recipe}. 
\begin{theorem}
\label{main_cc_th}
Let $P$ be a protocol found in the second step of Recipe~\ref{recipe},
and let $S$ be the subset of inputs picked in the last step.
Let $p=|S|/n^{k-1}$, $\gamma = C(P)$ and $N'=N\cdot 2^{\gamma}$, then

\begin{enumerate}

\item There is an (explicitly defined) $(k-1)$-graph on $n$ vertices whose edge density is $p$,
that is the union of $N'$ induced partial Steiner systems $S(k-2, k-1)$. 

\item If $p=1-o(1)$, then $h_{k-1}(n,c) \le (N'/n)^2$ for $c < d_k$. 
Here, the construction
of the $(k-1)$-graph that gives the bound is also explicit, given
explicit constructions of $(k-1)$-graphs of density
$d_k-o(1)$ which contain no $K_k$.

\end{enumerate}
\end{theorem}

We defer the proof of Theorem~\ref{main_cc_th} and the explicit definition of the graphs produced 
by Recipe~\ref{recipe} to Section~\ref{sec:recipe_proof_complete}. 
In the next section we give a simple example of how Theorem~\ref{main_cc_th} can be applied, then in 
Section~\ref{apply_recipe_to_prove_theorems} we apply it to prove Theorems~\ref{t11} and \ref{t12}.

\subsection{Applying Theorem~\ref{main_cc_th} - an example}
\label{an_example}

We apply Theorem~\ref{main_cc_th} to prove:
\begin{lemma}
\label{lem:simple_rs_graph}
There is a graph on $n$ vertices with edge density $\Omega(1/\log n)$ that is the union of
$n^{1+1/\Omega(\log\log n)}$ induced matchings.
\end{lemma}

\begin{proof}
We follow the steps of Recipe~\ref{recipe}:
\paragraph{Choosing the function}
Let $q,d > 1$ be natural numbers, denote $n=q^d$, and define $Z_{q, d} = \{\frac{1}{2}(x+y): x,y \in [q]^d\}$.
Denote by $g_{q, d}: ([q]^d)^2 \times Z_{q, d} \to \{0,1\}$ the function 
satisfying $g_{q, d}(x, y, z) = 1$ if and only if $x+y = 2z$ (here addition is
in $\mathbb{R}^d$). It is not hard to verify
that $g_{q, d}$ is a sub-permutation.
Denote $N = N_{q, d} = |Z_{q, d}|$, then
$$
N \le (2q)^d = q^d \cdot 2^d = n^{1+1/\log q}. 
$$ 
Since $\log \log n = \log d + \log\log q$, we have that $N \le n^{1+1/\Omega(\log\log n)}$ as long as $d \le q^c$
for some constant $c$. We will later choose $d = q^{4}$.

\paragraph{The protocol}

Next we present a protocol for $g_{q, d}$.

\bigskip
\fbox{
\begin{minipage}{5.8in}
\begin{protocol}{\sf{A protocol for $g_{q, d}$}}
\label{protocol_AP}
\begin{enumerate}

\item The $z$-player computes $\|x-y\|_2^2$, and writes the result on the board.

\item The $y$-player writes $1$ iff $\|x-y\|_2^2 = 4\|x-z\|_2^2$.

\item The $x$-player writes $1$ iff $\|x-y\|_2^2 = 4\|y-z\|_2^2$.

\end{enumerate}

\end{protocol}
\end{minipage}
}
\bigskip

At the end, all players know the value of the function. Indeed, the value of the function is $1$ 
if the last two bits written on the board are both equal to $1$, and $0$ otherwise.

\paragraph{The cost of the protocol}
The cost of the protocol is $C(P)=2$, as the first two players send only $2$ 
verification bits.

\paragraph{The choice of $S$}

By the Chernoff-Hoeffding's inequality (c.f., e.g., \cite{AS}), 
the quantity $\|x-y\|_2^2$ computed by the third player satisfies
$$
P(\left| \|x-y\|_2^2 - \mathbb{E}(\|x-y\|_2^2) \right| \ge t) \le 2e^{-\frac{2t^2}{dq^4}}.
$$
Thus, with constant probability, $\|x-y\|_2^2$ takes one of $\sqrt{d}q^2$ values. There is, therefore, a 
transcript $\trans$ for the third player such that $|S_3(\trans)| \ge \Omega(n^2/\sqrt{d}q^2)$.
If we take $d = q^{4}$ we get $|S_3(\trans)| \ge \Omega(n^2/d) \ge \Omega(n^2/\log n)$.
The fact that $S_3(\trans)$ is symmetric is easy to verify.
Lemma~\ref{lem:simple_rs_graph} now follows from 
Theorem~\ref{main_cc_th}, part 1.
\end{proof}

Note that we could improve the density of the graph 
in Lemma~\ref{lem:simple_rs_graph}
to $\Omega(\log \log n/ \log ^{\epsilon} n)$ for 
any constant $\epsilon > 1/2$ by taking
$d=q^c$ for an appropriately chosen large 
constant $c$. This seems to be the best
one can get when using Protocol~\ref{protocol_AP} though. In the next section we use a variant
of this protocol in which the first two players participate more, in order to save communication 
bits of the last player. This will allow us 
to increase the density to near optimal.

\section{Applying Theorem~\ref{main_cc_th} to prove Theorems~\ref{t11} and \ref{t12}}
\label{apply_recipe_to_prove_theorems}

\subsection{The case $k=3$}
\label{The_constructions_of_Alon_etal}

\paragraph{Choosing the function}

The function we choose is $g_{q,d}$, defined in Section~\ref{an_example}.
We later fix $d = q^5$.

\paragraph{The protocol}

For a natural number $r$ let $G_r=(V, E_r)$ be the graph 
with $V = [q]^d$, where $d$ is even, and 
$E_r = \{x,y : \|x-y\|_2^2 \le r\}$ (later we take $r=\sqrt{d}$). 
The players agree on a proper coloring $\chi$ of $G_{2r}$
by $d_{2r}+1$ colors, where $d_{2r}$ is its maximum degree. Let
$\mu = \mathbb{E}(\|x-y\|_2^2) = \frac{1}{6}d(q^2 - 1)$, the players 
also agree on some
partition $P$ of $[0, dq^2]$ into intervals of length $r^2+O(1)$. 
The players choose $P$
that satisfy: the number of intervals in the partition is $\lceil dq^2/r^2 \rceil$, and the number 
$\mu$ is in the middle of the interval containing it. As an example, the players can choose 
a partition which is a translation of the partition induced by $DIV(L) = \lfloor \frac{L}{r^2} \rfloor$.
Let $I_r: [0, dq^2] \to \{0,1,\ldots,dq^2/r^2\}$ map a number in $[0, dq^2]$ to the index of the interval 
containing it, according to $P$. 
Given an input $(x,y,z)$, the players then use the following protocol:

\bigskip
\fbox{
\begin{minipage}{5.8in}
\begin{protocol}{\sf{A protocol for $g_{q,d}$}}
\label{protocol_for_rs_2}
\begin{enumerate}

\item The $z$-player writes $I_r(\|x-y\|_2^2)$ on the board. 

\item The $y$-player verifies that $I_r(\|x-y\|_2^2) = I_r(4\|x-z\|_2^2)$,
and writes $1$ on the board iff this is the case.

\item The $x$-player verifies that $I_r(\|x-y\|_2^2) = I_r(4\|y-z\|_2^2)$,
and writes $1$ on the board iff this is the case.

\item If one of the last two bits are equal to $0$, reject and finish.

\item The $x$-player writes $\chi(2z-y)$ on the board.

\item The $y$-player writes the value of $g_{q,d}(x,y,z)$.

\end{enumerate}

\end{protocol}
\end{minipage}
}
\bigskip

\begin{theorem}
\label{protocol_2_correctness}
Protocol~\ref{protocol_for_rs_2} is correct.
\end{theorem}

For the proof of correctness, we use the following two observations (used also in \cite{AMS}):
\begin{lemma}[Parallelogram law]
Let $x,y,z \in \mathbb{R}^d$ then:
\[
\|x-y\|_2^2+\|x+y-2z\|_2^2 = 2\|x-z\|_2^2 + 2\|y-z\|_2^2
\]
\end{lemma}

\begin{lemma}[\cite{AMS}]
\label{deg}
For an even integer $d > 0$, the number of integral points contained in the ball of radius $r$ 
in $\mathbb{R}^d$ is at most:
\[
\frac{\pi^{d/2}(r+0.5)^d}{(d/2)!} < \frac{(2 \pi e)^{d/2}(r+0.5 \sqrt d)^d}{(d)^{d/2}}
\]
\end{lemma} 

\begin{proof}[of Theorem~\ref{protocol_2_correctness}]
By Lemma~\ref{deg}, the maximum degree of $G_r$ is at most
$$
d_r = \frac{(2 \pi e)^{d/2}(r+0.5 \sqrt d)^d}{(d)^{d/2}}.
$$
The chromatic number of $G_{2r}$ is therefore at most $d_{2r}+1$. 

If $x+y=2z$ then obviously the protocol reaches step 5. On the other hand,
if the protocol reached step 5 then $\|x-y\|_2^2$, $4\|x-z\|_2^2$, and $4\|y-z\|_2^2$,
all lie in the same interval of length $r^2$. Thus, by the Parallelogram law
\begin{eqnarray*}
\|x+y-2z\|_2^2 &=& 2\|x-z\|_2^2 + 2\|y-z\|_2^2-\|x-y\|_2^2 \\
&=& \frac{1}{2}\left( 4\|x-z\|_2^2 + 4\|y-z\|_2^2 \right) -\|x-y\|_2^2 \\
&\le& r^2.
\end{eqnarray*}
Thus, $(2z-y)$ is in a ball $B(x,r)$ of radius $r$ around $x$. Every other vector $v \in B(x,r)$ is in distance at most
$2r$ from $(2z-y)$. The color of $(2z-y)$ in this ball is therefore unique. It follows that at step 6 the $y$-player knows 
the value of $y$ and hence knows everything.
\end{proof}

\paragraph{The cost of the protocol}

The number of bits used by the first two players is:
\begin{eqnarray*}
\log d_{2r} + \Theta(1) &=& \Theta \left( d+d \log \frac{2r+0.5\sqrt{d}}{\sqrt{d}} \right).
\end{eqnarray*}
If we take $r=\sqrt{d}$, the cost of the protocol is therefore bounded by
$$
C(P) \le O(d) = O\left( \frac{\log n}{\log q} \right).
$$

\paragraph{The choice of $S$}

A transcript $\trans$ of the $z$-player corresponds to a message $I_r(\|x-y\|_2^2)$. 
The size of $S_3(\trans)$ is therefore equal to the number of 
pairs $x,y \in [q]^d$ satisfying $I_r(\|x-y\|_2^2) = \trans$. Hoeffding's inequality 
implies that 
$$
P(\left| \|x-y\|_2^2 - \mu \right| \ge t) \le 2e^{-\frac{2t^2}{dq^4}}.
$$
In particular, the probability that $I_r(\|x-y\|_2^2) = I_r(\mu)$ is at least
$(1- 2e^{-\frac{r^4}{2dq^4}})$ since we chose the partition of the intervals so that $\mu$ 
lies in the middle of the interval containing it.

Take $r = \sqrt{d}$, and pick $S=S_3(\trans)$ for
$\trans=I_{\sqrt{d}}(\mu)$, we have 
\[
|S| \ge (1- 2e^{-\frac{d}{2q^4}})n^2.
\]

\paragraph{Conclusion}
When applying Theorem~\ref{main_cc_th} the parameters that we get are:
\begin{itemize}

\item $p = (1- 2e^{-\frac{d}{2q^4}})$,

\item $N' = n^{1+1/\Omega(\log \log n)} 2^{O(d)}$. 

\end{itemize}
Taking $d=q^5$, and observing that $S$ is symmetric, this proves the $k=3$ case of Theorems~\ref{t11} 
and \ref{t12}.

\subsection{The case $k > 3$}

\paragraph{Choosing the function} 

Let $Z_{m,q,d} = \{\frac{1}{m}(\sum_{i=1}^m x_i): x_i \in [q]^d\}$
and define $g_{k,q,d}: ([q]^d)^{k-1} \times Z_{k-1,q,d} \to \{0,1\}$ 
by $g_{k,q,d}(x_1,\ldots,x_k) = 1$ if and only if $x_1+\cdots+x_{k-1} = (k-1)x_k$.
It is easy to verify that $g_{k,q,d}$ is a sub-permutation, and
$$
|Z_{k-1,q,d}| \le (kq)^d = n^{1+1/\log_k q}. 
$$

\paragraph{The protocol} The protocol is a simple reduction to the case $k=3$.

\bigskip
\fbox{
\begin{minipage}{5.8in}
\begin{protocol}{\sf{A protocol for $g_{k,q,d}$}}
 
\begin{enumerate}

\item The first player writes $1$ on the board if and only if $\frac{1}{2} ((k-1)x_k - x_3 - \cdots - x_{k-1}) \in Z_{2,q,d}$.

\item If the last bit was equal to $0$, the protocol ends with rejection.

\item Players $1$, $2$ and $k$ run Protocol~\ref{protocol_for_rs_2} for $g_{3,q,d}$ 
with $r = \sqrt{d}$ on input $x' = x_1, y' = x_2$, and $z'=\frac{1}{2} ((k-1)x_k - x_3 - \cdots - x_{k-1})$.

\end{enumerate}

\end{protocol}
\end{minipage}
}
\bigskip

The correctness of the above protocol follows from the correctness of Protocol~\ref{protocol_for_rs_2} and the fact that the equation
$x_1 + x_2 + x_3 + \cdots + x_{k-1} = (k-1)x_k$ holds if and only if $x_1 + x_2 = 2(\frac{1}{2} ((k-1)x_k - x_3 - \cdots - x_{k-1}))$.
Note that the last equation cannot hold if $\frac{1}{2} ((k-1)x_k - x_3 - \cdots - x_{k-1})$ does not belong to $Z_{2,q,d}$.

\paragraph{The cost of the protocol}

Outside the reduction to Protocol~\ref{protocol_for_rs_2}, the players send only one more bit.
The cost of the protocol thus satisfy $C(P) \le O(d) \le O(\frac{\log n}{\log q})$, as before.

\paragraph{The choice of $S$} We can choose, as in Section~\ref{The_constructions_of_Alon_etal},
the set $S=S_k(\trans)$ for $\trans = I_{\sqrt{d}}(\mu)$.
By Hoeffding's inequality, the size of $S$ is $(1-o(1))n^{k-1}$ as long as $d >> q^4$. 
The only problem is that $S$ is not symmetric.
To remedy that, just add to the protocol a test whether 
$I_r(\|x_i-x_j\|_2^2) = I_r(\mu)$ for every $1 \le i < j < k$. 
These tests can all be carried out by the last player, so this adds only one
more communication bit, which for simplicity we assume is the last bit.
Now pick the transcript $\trans' = (\trans, 1)$ which imply that $I_r(\|x_i-x_j\|_2^2) = I_r(\mu)$
for all $1 \le i < j < k$. The corresponding set $S_k(\trans')$ is now symmetric, and as long as $k$
is a constant, Hoeffding's inequality still implies that the size of $S_k(\trans')$ is at least $(1-o(1))n^{k-1}$.

\section{Proof of Theorem~\ref{main_cc_th}}
\label{sec:recipe_proof_complete}

We first rephrase Theorem~\ref{main_cc_th} slightly.
\begin{theorem}
\label{main_cc_th_2}
Let $f:[n]^{k-1}\times[N] \to \{0,1\}$ be a weak sub-permutation, 
and let $S$ be a symmetric cylinder intersection (w.r.t. $f$). 
Let $p=|S|/n^{k-1}$, then

\begin{enumerate}

\item There is an (explicitly defined) $(k-1)$-graph on $n$ vertices whose edge density is $p$,
that is the union of $N$ induced partial Steiner systems $S(k-2, k-1)$. 

\item If $p=1-o(1)$, then $h_{k-1}(n,c) \le (N/n)^2$ for $c < d_k$. 
Here, the construction
of the $(k-1)$-graph that gives the bound is explicit, given
explicit constructions of $(k-1)$-graphs of density
$d_k-o(1)$ which contain no $K_k$.

\end{enumerate}
\end{theorem}

\begin{lemma}
Theorem~\ref{main_cc_th_2} implies Theorem~\ref{main_cc_th}.
\end{lemma}

\begin{proof}
The difference between 
Theorem~\ref{main_cc_th_2} and Theorem~\ref{main_cc_th} lies in the different properties
of the subset $S$. In Theorem~\ref{main_cc_th} $S$ is defined by
\[
S=S_k(\trans_k) = \{ (x_1,\ldots,x_k) \in [n]^{k-1} \times [N]: \mathcal{T}_k(x_1,\ldots,x_k) = \trans_k 
\;\; and \;\; f(x_1,\ldots,x_k)=1 \}, 
\]
for some transcript $\trans_k$ of {\bf the last player}. In Theorem~\ref{main_cc_th_2} on the other hand,
$S$ is a cylinder intersection, that is
\[
S =S(\trans)= \{ (x_1,\ldots,x_k) \in [n]^{k-1} \times [N]: \mathcal{T}(x_1,\ldots,x_k) = \trans 
\;\; and \;\; f(x_1,\ldots,x_k)=1 \}, 
\]
for some transcript $\trans$ of {\bf all players}.

This difference is easily bridged though. Let $f:[n]^{k-1} \times [N] \to \{0,1\}$ be a weak sub-permutation,
$P$ a protocol for $f$, $\trans_k$ a transcript of the last player, and $S=S_k(\trans_k)$ a subset, 
found using Recipe~\ref{recipe}. Let $\gamma = C(P)$,
and denote $N' = N\cdot 2^{\gamma}$. For simplicity identify $[N']$ with $[N]\times \{0,1\}^{\gamma}$.

Define $g:[n]^{k-1} \times [N'] \to \{0,1\}$ by $g(x_1,\ldots,x_{k-1}, (x_k, \trans_{1\ldots k-1}))=1$ 
if and only if $f(x_1,\ldots,x_{k-1}, x_k)=1$
and $\trans_{1\ldots k-1} = \mathcal{T}_1(x_1,\ldots, x_k) \circ \cdots \circ \mathcal{T}_{k-1}(x_1,\ldots, x_k)$.
That is, $\trans_{1\ldots k-1}$ is the message written on the board by the first $k-1$ players, according to protocol $P$,
on input $(x_1,\ldots , x_k)$. 

It is not hard to verify that $g$ is a weak sub-permutation. 
We use the following protocol $P'$ for $g$, on input $(x_1,\ldots,x_{k-1}, (x_k, \trans_{1\ldots k-1}))$: 
the last player sends his message as in $P$, 
then each of the other players verifies (using one bit 
of communication each) that his part in $\trans_{1\ldots k-1}$ agrees with $P$.  
Obviously $P'$ is correct if and only if $P$ is correct. 
The subset
\[
S' = \{ (x_1,\ldots,(x_k, \trans_{1\ldots k-1})) \in [n]^{k-1} \times [N']: \mathcal{T}_k(x_1,\ldots,x_k) = \trans_k 
\;\; and \;\; f(x_1,\ldots,x_k)=1 \}
\]
is a cylinder intersection with respect to $P'$ and $g$, and $|S'|/n^{k-1} = |S|/n^{k-1}$. 
Theorem~\ref{main_cc_th_2} can now be applied to prove Theorem~\ref{main_cc_th}.
\end{proof}

In the rest of this section we prove Theorem~\ref{main_cc_th_2}. For simplicity
we first prove it for the case of graphs ($k=3$) and then explain the necessary adjustments
for the general case ($k \ge 3$).

\subsection{The case $k=3$}
\label{sec:recipe_proof_complete_k_3}

We prove the first conclusion of Theorem~\ref{main_cc_th_2}, concerning
\rs\ graphs, in Section~\ref{first_conclusion}. The upper bound on $h(n,c)$ is proved in Section~\ref{second_conclusion}.
We use the following simple fact proved in \cite{hdp17}.
\begin{lemma}[\cite{hdp17}]
\label{cylinders_and_corners}
Let $f:[n]\times[n]\times[N] \to \{0,1\}$ be a function satisfying
that every line in the third dimension
contains at most a single $1$, and let $S$ be a cylinder intersection 
(w.r.t $f$).
Then, $S$ does not contain {\em stars}: triplets of the form $(x',y,z),(x,y',z),(x,y,z')$
where $x\ne x'$, $y\ne y'$ and $z\ne z'$.
\end{lemma}

\subsubsection{\rs\ graphs}
\label{first_conclusion}

The relation between \rs\ graphs and the communication complexity of $2$-dimensional permutations
was observed in \cite{hdp17}. The graphs constructed in \cite{hdp17} 
are bipartite though, and we need slightly 
different settings.
Let $S \subseteq [n]\times[n]\times[N]$ be symmetric, define
\[
E_S = \{(x,y),(x,z),(y,z): (x,y,z)\in S\}.
\]
Let $G_S = (V, E_S)$ be the graph with vertex set $V=V_A\cup V_B$,
where $V_A=[n]$ and $V_B=[N]$, and edge set $E_S$. We allow self loops in $E_S$,
and consider a collection of self loops as a matching. Note that when $S$ is a cylinder intersection 
with respect to a weak sub-permutation
there is always at most one edge between a pair of vertices. 
The following lemma implies the first conclusion in Theorem~\ref{main_cc_th_2}.

\begin{lemma}
\label{cor_1}
Let $f:[n]\times[n]\times[N] \to \{0,1\}$ be a weak sub-permutation, 
and let $S$ be a symmetric cylinder intersection. 
Let $H = ([n], F)$ be the subgraph of $G_S$ induced on $V_A$. That is:
\[
F = \{(x,y) : \exists z \in V_B \text {  s.t.  }(x,y,z) \in S\}.
\]
Then, the edges of $|F|$ can be partitioned into $N$ induced matchings.
\end{lemma}

\begin{proof}
Partition the edge set $F$ as follows, for every $z\in B$ let
$$
F_z = \{ (x,y) : (x,y,z)\in S\}.
$$  
This is a partition of $F$ since $f$ a sub-permutation, and therefore
there is at most a single $z$ such that $(x,y,z) \in S$ for every $(x,y) \in [n]^2$.

The fact that $F_z$ is an induced matching follows from Lemma~\ref{cylinders_and_corners}.
Assume in contradiction that $F_z$ is not an induced matching, then there
is an edge $(x, y) \in F_{z'}$ for $z' \ne z$ such that $(x,y'), (x',y)$
are in $F_z$. We then get a star $(x',y,z), (x,y',z), (x,y,z') \in S$, 
contradicting Lemma~\ref{cylinders_and_corners}. Note that the fact that $f$ 
is a sub-permutation also implies that $x' \ne x$ and $y' \ne y$.    
\end{proof}

\subsubsection{An upper bound on $h(n,c)$}
\label{second_conclusion}

Consider the same graph $G_S$ as in the previous section. A basic observation is:
\begin{lemma}
\label{main}
Let $f:[n]\times[n]\times[N] \to \{0,1\}$ be a function satisfying
that every line in the third dimension
contains at most a single $1$, and let $S$ 
be a symmetric cylinder intersection (w.r.t $f$).
Then, a triangle $(x,y,z)$
where $x,y \in V_A$ and $z \in V_B$ exists
in $G_S$ if and only if $(x,y,z) \in S$.
\end{lemma}

\begin{proof}
The fact that a triangle $(x,y),(x,z),(y,z)$ where $x,y \in V_A$ and $z \in V_B$ exists
in $G_S$ for every $(x,y,z) \in S$ follows immediately from the definition of $E_S$.
Assume in contradiction that there is also such a triangle
in $G_S$ for $(x,y,z) \not \in S$. Then necessarily there 
are $x', y' \in V_A$  and $z' \in V_B$
such that $(x', y, z), (x, y', z), (x, y, z') \in S$. But then $S$ contains a star, in contradiction 
to Lemma~\ref{cylinders_and_corners}.
\end{proof}

\begin{lemma}
\label{cor_2}
Let $f:[n]\times[n]\times[N] \to \{0,1\}$ be a weak sub-permutation,
and let $S$ be a symmetric cylinder intersection satisfying $|S| = (1-o(1))n^2$. 
Then $h(n,c) \le N^2/n^2$ for $c < 1/4$.
\end{lemma}

\begin{proof}
Consider the graph $G_S$ again. 
By lemma~\ref{main}, and the fact that $f$ is a weak sub-permutation,
an edge in $G_S$ appears in exactly one triangle $(x,y,z)$ with
$x,y \in V_A$ and $z \in V_B$. Therefore, if we take a bipartite subgraph inside $V_A$, 
we will have every edge lie in exactly one triangle, which is optimal. 
But, the density of edges in $G_S$ is relatively small, since there are $n+N$
vertices and order of $(1-o(1))n^2$ edges.
To remedy this, we define a product function, aiming to increase the density
of edges. The price we pay is that the number of triangles an edge can lie in
increases.

Let $t \ge 2$ be a natural number, define $f^t:([2^t]\times[n])^2 \times [N] \to \{0,1\}$
by $f((\alpha,x),(\beta,y),z)=1$ if and only if $f(x,y,z)=1$. Let 
\[
S^t = \{ ((\alpha,x),(\beta,y),z) : (x,y,z) \in S\}.
\]
It is not hard to verify that $S^t$ is a symmetric cylinder
intersection with respect to $f^t$.
By Lemma~\ref{main} a triangle $((\alpha, x),(\beta, y),z)$ where $(\alpha, x),(\beta, y) \in ([2^t]\times[n])$ and $z \in [N]$ exists
in $G_{S^t}$ if and only if $(x,y,z) \in S$. Thus, every edge of $G_{s^t}$ lies in at most $2^t$ triangles of this sort. To remove other kind of triangles
let $H=([2^t]\times[n], E_H)$ be a bipartite graph with density $1/4$. Now define
\[
E'_{S^t} = \{((\alpha,x),(\beta,y)), ((\alpha,x), z), ((\beta,y), z): (x,y,z)\in S,  \;\; ((\alpha,x), (\beta,y)) \in E_H \}.
\] 
Then every edge in $E'_{S^t}$ lies in at least one triangle and at most $2^t$ triangles.
The number of edges satisfy $|E'_{S^t}| \ge (1-o(1))(2^tn)^2/4$. The density of edges 
is thus 
\[
(1-o(1))\frac{1}{4}\frac{(2^tn)^2}{(2^tn+N)^2}.
\]
If we take $t=2\log (N/n)$ this becomes
\[
(1-o(1))\frac{1}{4}\frac{(N^2/n)^2}{(N^2/n+N)^2}.
\]
Recall that $S$ is a 
cylinder intersection of size $(1-o(1))n^2$.
It therefore follows from the graph removal lemma (and the
hypergraph removal lemma for larger $k$) - 
see Theorem~34 in \cite{hdp17} for details -
that necessarily $n = o(N)$. 
The density is thus $(1-o(1))\frac{1}{4}$. Since 
every edge is in at most $2^t = N^2/n^2$ triangles, this completes
the proof.
\end{proof}

\subsection{The general case}
\label{hypergraphs}

We outline the proof of Theorem~\ref{main_cc_th_2} for $k \ge 3$. Since the general case is very similar to 
the proof of the $k=3$ case, we do not repeat all the details here. 

For $\vec{x}=(x_1,\ldots,x_k) \in [n]^{k-1}\times [N]$ denote by $[\vec{x}]_{k-1}$ the family of all
subsets of size $k-1$ of entries of $\vec{x}$. 
That is:
$$
[\vec{x}]_{k-1} = \binom{\{x_1,\ldots,x_k\}}{k-1}.
$$
Let $S \subseteq [n]^{k-1}\times[N]$ be a symmetric subset of entries, define
\[
E_S =  \bigcup_{\vec{x} \in S} [\vec{x}]_{k-1}.
\]

Let $G_S = (V, E_S)$ be the $(k-1)$-graph with vertex set $V=V_A\cup V_B$,
where $V_A=[n]$ and $V_B=[N]$, and edge set $E_S$. 

The generalized version of Lemma~\ref{cylinders_and_corners} is:
\begin{lemma}[\cite{hdp17}]
\label{cylinders_and_corners_general}
Let $f:[n]^{k-1}\times[N] \to \{0,1\}$ be a function satisfying
that every line in the $k$th dimension
contains at most a single $1$, and let $S$ be a cylinder intersection (w.r.t $f$).
Then, $S$ does not contain {\em stars}: $k$ entries of the form 
$(x'_1,x_2,\ldots,x_k),(x_1,x'_2,\ldots,x_k),(x_1,x_2,\ldots,x'_k)$
where $x'_i\ne x_i$ for $i=1\ldots k$.
\end{lemma}
This immediately gives:
\begin{lemma}
\label{main_hyper}
Let $f:[n]^{k-1}\times[N] \to \{0,1\}$ be a function satisfying
that every line in the $k$th dimension
contains at most a single $1$, and let $S$ be a symmetric cylinder intersection (w.r.t $f$).
Then, we have that $[\vec{x}]_{k-1}$ is a copy of $K_k$ in $G_S$ with
$x_1,\ldots x_{k-1} \in V_A$ and $x_k \in V_B$,
if and only if $\vec{x}=(x_1,\ldots,x_k) \in S$.
\end{lemma}

\begin{proof}
Similar to the proof of Lemma~\ref{main}, but using Lemma~\ref{cylinders_and_corners_general}
instead of Lemma~\ref{cylinders_and_corners}.
\end{proof}

The following two lemmas generalize Lemma~\ref{cor_1} and Lemma~\ref{cor_2}:
\begin{lemma}
\label{cor_1_general}
For an integer $k\ge 3$, let $f:[n]^{k-1}\times[N] \to \{0,1\}$ be a weak sub-permutation, 
and let $S$ be a symmetric cylinder intersection. 
Let $G' = ([n], E')$ be the subrgraph of $G_S$ induced on $V_A$.
Then, the edges of $|E'|$ can be partitioned into $N$ partial Steiner systems $S(k-2,k-1)$.
\end{lemma}

\begin{proof}
The proof is similar to the proof of Lemma~\ref{cor_1}, we rewrite the main points. The edges of $G'$ are:
\[
E' = \{ (x_1,\ldots,x_{k-1}) : \exists x_k \in V_B \;\; s.t \;\; (x_1,\ldots,x_{k-1},x_k) \in S\}.
\]
Partition the edge set $E'$ as follows, for every $z\in V_B$ let
$$
E'_z = \{ (x_1,\ldots,x_{k-1}) : (x,\ldots,x_{k-1},z)\in S\}.
$$  
This is a partition of $E'$ since $f$ a sub-permutation, and 
the fact that $E'_z$ is a partial Steiner system follows from Lemma~\ref{cylinders_and_corners_general}.
\end{proof}

\begin{lemma}
\label{cor_2_general}
For an integer $k\ge 3$, let $f:[n]^{k-1}\times[N] \to \{0,1\}$ be a weak sub-permutation,
and let $S$ be a symmetric cylinder intersection satisfying $|S| = (1-o(1))n^{k-1}$. 
Then $h_{k-1}(n,c) \le (N/n)^2$ for $c < d_k$.
\end{lemma}

\begin{proof}
The proof is very similar to 
the proof of Lemma~\ref{cor_2}, just instead of taking
the subgraph $H=([2^t]\times[n], E_H)$ to 
be a bipartite graph with density $1/4$, take a subhypergraph
with no copies of $K_k$ and density $d_k$. Note that we do not need to know $d_k$ or $H$, we just
need to know that $d_k$ is finite and that $H$ exists.  
\end{proof}

\section{Summary}
\label{s5}

As mentioned in the introduction, there is a link between the main construction of \cite{AMS}
and the original construction of Ruzsa and Szemer\'{e}di \cite{RSz}. We describe this link here,
starting with a new construction, equivalent to the one of Ruzsa and Szemer\'{e}di, derived
using the recipe in Section~\ref{subsec:a_recipe}. Our approach avoids the use of Behrends construction 
of a large set of integers without a three-term arithmetic progression \cite{Be}, which was the heart
of the construction of Ruzsa and Szemer\'{e}di.

\begin{lemma}[\cite{RSz}]
There exists a graph on $n$ vertices, with $n^2/2^{O(\sqrt{\log n})}$ edges, that is the union
of $\Theta(n)$ induced matchings.
\end{lemma}

\begin{proof}
We follow the steps of Recipe~\ref{recipe}. The details are very similar to those in Section~\ref{an_example},
with slight modifications.

\paragraph{Choosing the function}
Let $q,d > 1$ be natural numbers and denote $n=q^d$. 
Let $f_{q, d}: ([q]^d)^3 \to \{0,1\}$ be the function 
satisfying $f_{q, d}(x, y, z) = 1$ if and only if $x+y = 2z$. 
It is not hard to verify
that $f_{q, d}$ is a weak sub-permutation, in fact it is a weak permutation.
We later set $q$ to be even and $d=\log(q)=\Theta(\sqrt{\log n})$.

\paragraph{The protocol}

The protocol is identical to the protocol for $g_{q, d}$ in Section~\ref{an_example}.

\paragraph{The cost of the protocol}
The cost of the protocol is $C(P)=2$.

\paragraph{The choice of $S$}

By Hoeffding's inequality, 
with constant probability, $\|x-y\|_2^2$ takes one of $\sqrt{d}q^2$ values. There is, therefore, a 
transcript $\trans$ for the third player such that $|S_k(\trans)| \ge \Omega(|f_{q,d}^{-1}(1)|/\sqrt{d}q^2)$.
Where $|f_{q,d}^{-1}(1)|$ is the number of $1$'s of the function $f_{q,d}$. That is, it is the
number of $x,y \in [q]^d$ such that $(x+y)/2$ is also in $[q]^d$. Assume for simplicity that $q$ is even,
then $|f_{q,d}^{-1}(1)| \ge q^d \cdot (q/2)^d$. Therefore
$$
|S_k(\trans)| \ge \Omega(q^d \cdot (q/2)^d/\sqrt{d}q^2) \ge \Omega(n^2/2^d\sqrt{d}q^2).
$$ 
Taking $d = \log q = \Theta(\sqrt{\log n})$ we get $|S_k(\trans)| \ge n^2/2^{O(\sqrt{\log n})}$.
$S_k(\trans)$ is symmetric, thus Lemma~\ref{lem:simple_rs_graph} follows from Theorem~\ref{main_cc_th}.
\end{proof}

We can now describe the relation between the construction of Ruzsa and Szemer\'{e}di \cite{RSz}
and that of \cite{AMS}. Call the construction above $A$, the simple construction of Section~\ref{an_example} $B$,
and the construction of 
Section~\ref{The_constructions_of_Alon_etal} (providing
the  graphs similar to \cite{AMS}) $C$. 
The table below compares these constructions. 

\begin{center}
\begin{tabular}{|p{2.5cm}|p{3.5cm}|p{3.5cm}|p{3.5cm}|}
\hline
 & A & B & C \\
\hline
Function                    &  Domain: $([q]^d)^3$ \;\;\;\;\; Def. rule: x+y=2z & 
Dom.: $([q]^d)^2 \times Z_{q,d}$ \;\;\;\;\; Def. rule: x+y=2z  & 
Dom.: $([q]^d)^2 \times Z_{q,d}$ \;\;\;\;\; Def. rule: x+y=2z\\
\hline
Protocol idea              &  Third player sends $\|x-y\|_2^2$. & 
Third player sends $\|x-y\|_2^2$. & 
Third player sends some bits of $\|x-y\|_2^2$, then the first two players compute the rest.\\
\hline
Number of vertices & $n=q^d$ & $n=q^d$ & $n=q^d$ \\
\hline
Edge density                     &  $2^{-O(\sqrt{\log n})}$ & 
$\Omega(\log \log n/ \log ^{\epsilon} n)$ for any constant $\epsilon > 1/2$ & 
1-o(1)\\
\hline 
Number of matchings  & $\Theta(n)$ &
$n^{1+O(1/\log\log n)}$ &
$n^{1+O(1/\log\log n)}$ \\
\hline
\end{tabular}
\end{center}

\end{document}